\documentclass[11pt]{article}
\usepackage{amsmath,amsthm,amsfonts,amssymb}

\allowdisplaybreaks

\DeclareMathOperator{\Hom}{Hom}
\DeclareMathOperator{\Com}{Com}

\DeclareMathOperator{\II}{I}
\DeclareMathOperator{\1}{id}

\newcommand{\NN}{\mathbb{N}}
\newcommand{\RR}{\mathbb{R}}

\newcommand{\EEnd}{\mathcal End}
\newcommand{\EE}{\mathcal E}
\newcommand{\bul}{\bullet}

\renewcommand{\=}{\doteq}
\renewcommand{\t}{\otimes}

\newcommand{\al}{\alpha}
\newcommand{\be}{\beta}
\newtheorem{thm}{Theorem}[section]
 \newtheorem{lemma}[thm]{Lemma}

\theoremstyle{definition}
 \newtheorem{defn}[thm]{Definition}
\theoremstyle{definition}
 \newtheorem{exam}[thm]{Example}
\theoremstyle{definition}
 
\begin{document}

\thispagestyle{empty}

\title{\LARGE\bf Note on 2d binary operadic \\ harmonic oscillator}
\author{Eugen Paal and J\"{u}ri Virkepu}
\author{\Large Eugen Paal and J\"{u}ri Virkepu\\ \\
Department of Mathematics, Tallinn University of Technology\\
Ehitajate tee 5, 19086 Tallinn, Estonia\\ \\
E-mails: eugen.paal@ttu.ee and jvirkepu@staff.ttu.ee}
\date{}

\maketitle
\thispagestyle{empty}

\begin{abstract}
It is explained how the time evolution of the operadic variables may be introduced. As an example, a $2$-dimensional binary operadic Lax representation of the  harmonic oscillator  is found.
\par\smallskip
{\bf 2000 MSC:} 18D50, 70G60
\par
{\bf Keywords:} Operad, harmonic oscillator, operadic Lax pair
\end{abstract}

\section{Introduction}

It is well known that quantum mechanical observables are linear \emph{operators}, i.e the linear maps $V\to V$ of a vector space $V$ and their time evolution is given by the Heisenberg equation. As a variation of this one can pose the following question \cite{Paal07}: how to describe the time evolution of the  linear algebraic operations (multiplications) $V^{\t n}\to V$. The algebraic operations (multiplications) can be seen as an example of the \emph{operadic} variables \cite{Ger,GGS92,KP,KPS}.

When an operadic system depends on time one can speak about \emph{operadic dynamics} \cite{Paal07}. 
The latter may be introduced by simple and natural analogy with the Hamiltonian dynamics.
In particular, the time evolution of operadic variables may be given by operadic Lax equation.
In \cite{PV07} it was shown how the dynamics may be introduced in 2d Lie algebra. 
In the present paper, an operadic Lax representation for harmonic oscillator is constructed
in general 2d binary algebras.

\section{Operad}

Let $K$ be a unital associative commutative ring, and let $C^n$ ($n\in\NN$) be unital $K$-modules.
For $f\in C^n$, we refer to $n$ as the \emph{degree} of $f$ and often write
(when it does not cause confusion) $f$ instead of $\deg f$. For
example, $(-1)^f\=(-1)^n$, $C^f\=C^n$ and $\circ_f\=\circ_n$. Also, it
is convenient to use the \emph{reduced} degree $|f|\=n-1$.
Throughout this paper, we assume that $\t\=\t_K$.

\begin{defn}[operad (e.g \cite{Ger,GGS92})]
A linear (non-symmetric) \emph{operad} with coefficients in $K$ is a sequence $C\=\{C^n\}_{n\in\NN}$ of unital
$K$-modules (an $\NN$-graded $K$-module), such that the following
conditions are held to be true.
\begin{enumerate}
\item[(1)]
For $0\leq i\leq m-1$ there exist \emph{partial compositions}
\[
  \circ_i\in\Hom(C^m\t C^n,C^{m+n-1}),\qquad |\circ_i|=0
\]
\item[(2)]
For all $h\t f\t g\in C^h\t C^f\t C^g$,
the \emph{composition (associativity) relations} hold,
\[
(h\circ_i f)\circ_j g=
\begin{cases}
    (-1)^{|f||g|} (h\circ_j g)\circ_{i+|g|}f
                       &\text{if $0\leq j\leq i-1$},\\
    h\circ_i(f\circ_{j-i}g)  &\text{if $i\leq j\leq i+|f|$},\\
    (-1)^{|f||g|}(h\circ_{j-|f|}g)\circ_i f
                       &\text{if $i+f\leq j\leq|h|+|f|$}.
\end{cases}
\]
\item[(3)]
Unit $\II\in C^1$ exists such that
\[
\II\circ_0 f=f=f\circ_i \II,\qquad 0\leq i\leq |f|
\]
\end{enumerate}
\end{defn}

In the second item, the \emph{first} and \emph{third} parts of the
defining relations turn out to be equivalent.

\begin{exam}[endomorphism operad \cite{Ger}]
\label{HG} Let $V$ be a unital $K$-module and
$\EE_V^n\={\EEnd}_V^n\=\Hom(V^{\t n},V)$. Define the partial compositions
for $f\t g\in\EE_V^f\t\EE_V^g$ as
\[
f\circ_i g\=(-1)^{i|g|}f\circ(\1_V^{\t i}\t g\t\1_V^{\t(|f|-i)}),
         \qquad 0\leq i\leq |f|
\]
Then $\EE_V\=\{\EE_V^n\}_{n\in\NN}$ is an operad (with the unit $\1_V\in\EE_V^1$) called the
\emph{endomorphism operad} of $V$.

Therefore, algebraic operations can be seen as elements of an endomorphism operad.
\end{exam}

Just as elements of a vector space are called \emph{vectors},  it is natural to call elements of an abstract operad \emph{operations}. The endomorphism operads can be seen as the most suitable objects for modelling operadic systems.

\section{Gerstenhaber brackets and operadic Lax pair}

\begin{defn}[total composition \cite{Ger,GGS92}]
The \emph{total composition} $\bul\:C^f\t C^g\to C^{f+|g|}$ is defined by
\[
f\bul g\=\sum_{i=0}^{|f|}f\circ_i g\in C^{f+|g|},
\qquad |\bul|=0
\]
The pair $\Com C\=\{C,\bul\}$ is called the \emph{composition algebra} of $C$.
\end{defn}

\begin{defn}[Gerstenhaber brackets \cite{Ger,GGS92}]
The  \emph{Gerstenhaber brackets} $[\cdot,\cdot]$ are defined in $\Com C$ as a graded commutator by
\[
[f,g]\=f\bul g-(-1)^{|f||g|}g\bul f=-(-1)^{|f||g|}[g,f],\qquad|[\cdot,\cdot]|=0
\]
\end{defn}

The \emph{commutator algebra} of $\Com C$ is denoted as $\Com^{-}\!C\=\{C,[\cdot,\cdot]\}$.
One can prove that $\Com^-\!C$ is a \emph{graded Lie algebra}. The Jacobi
identity reads
\[
(-1)^{|f||h|}[[f,g],h]+(-1)^{|g||f|}[[g,h],f]+(-1)^{|h||g|}[[h,f],g]=0
\]

Assume that $K\=\RR$ and operations are differentiable.
The dynamics in operadic systems (operadic dynamics) may be introduced by the

\begin{defn}[operadic Lax pair \cite{Paal07}]
Allow a classical dynamical system to be described by the evolution equations
\[
\dfrac{dx_i}{dt}=f_i(x_1,\dots,x_n),\quad i=1,\dots,n
\]
An \emph{operadic Lax pair} is a pair $(L,M)$ of homogeneous operations $L,M\in C$,
such that the above system of evolution equations is equivalent to the
\emph{operadic Lax equation}
\[
\dfrac{dL}{dt}=[M,L]\=M\bul L-(-1)^{|M||L|}L\bul M
\]
Evidently, the degree constraint $|M|=0$ gives rise to ordinary Lax pair \cite{Lax68,BBT03}.
\end{defn}

\section{Operadic harmonic oscillator}

Consider the Lax pair for the harmonic oscillator:
\[
L=\begin{pmatrix}
p&\omega q\\
\omega q &-p
\end{pmatrix},
\qquad
M=\frac{\omega}{2}
\begin{pmatrix}
0&-1\\
1&0
\end{pmatrix}
\]
Since the Hamiltonian is
\[
H(q,p)=\frac{1}{2}(p^2+\omega^2q^2)
\]
it is easy to check that the Lax equation
\[
\dot{L}=[M,L]\= ML - LM
\]
is equivalent to the Hamiltonian system
\[
\dfrac{dq}{dt}=\dfrac{\partial H}{\partial p}=p,
\quad
\dfrac{dp}{dt}=-\dfrac{\partial H}{\partial q}=-\omega^2q
\]
If $\mu$ is a homogeneous operadic variable one can use the above Hamilton's equations to obtain
\[
\dfrac{d\mu}{dt}
=\dfrac{\partial\mu}{\partial q}\dfrac{dq}{dt}+\dfrac{\partial\mu}{\partial p}\dfrac{dp}{dt}
=p\dfrac{\partial\mu}{\partial q}-\omega^2q\dfrac{\partial\mu}{\partial p}
\]
Therefore, the linear partial differential equation for the operadic variable $\mu(q,p)$ reads
\[
p\dfrac{\partial\mu}{\partial q}-\omega^2q\dfrac{\partial\mu}{\partial p}=M\bul\mu- \mu\bul M
\]
By integrating one gains sequences of operations  called the \emph{operadic (Lax representations of) harmonic oscillator}.

\section{Example}

Let $A\=\{V,\mu\}$ be a  binary algebra with operation $xy\=\mu(x\t y)$.
We require that $\mu=\mu(q,p)$ so that $(\mu,M)$ is an operadic Lax pair, i.e the operadic Lax equation
\[
\dot{\mu}=[M,\mu]\= M\bul\mu-\mu\bul M,\qquad |\mu|=1,\quad |M|=0
\]
is equivalent to the Hamiltonian system of the harmonic oscillator.

Let $x,y\in V$. By assuming that $|M|=0$ and $|\mu|=1$, one has
\begin{align*}
M\bul\mu
&=\sum_{i=0}^0(-1)^{i|\mu|}M\circ_i\mu
=M\circ_0\mu=M\circ\mu\\
\mu\bul M &=\sum_{i=0}^1(-1)^{i|M|}\mu\circ_i M =\mu\circ_0
M+\mu\circ_1 M=\mu\circ(M\t\1_V)+\mu\circ(\1_V\t M)
\end{align*}
Therefore, one has
\[
\dfrac{d}{dt}(xy)=M(xy)-(Mx)y-x(My)
\]
Let $\dim V=n$.
In a basis $\{e_1,\ldots,e_n\}$ of $V$,  the structure constants $\mu_{jk}^i$ of $A$ are defined by
\[
\mu(e_j\t e_k)\= \mu_{jk}^i e_i,\qquad j,k=1,\ldots,n
\]
In particular,
\[
\dfrac{d}{dt}(e_je_k)=M(e_je_k)-(Me_j)e_k-e_j(Me_k)
\]
By denoting $Me_i\= M_i^se_s$, it follows that
\[
\dot{\mu}_{jk}^i=\mu_{jk}^sM_s^i-M_j^s\mu_{sk}^i-M_k^s\mu_{js}^i,\qquad i,j,k=1,\ldots, n
\]
In particular, one has
\begin{lemma}
\label{lemma:first}
Let $\dim V=2$ and 
$
M\=(M_j^i)\=
\frac{\omega}{2}
\left(
\begin{smallmatrix}
0&-1\\
1&0
\end{smallmatrix}
\right)
$.
Then the $2$-dimensional binary operadic Lax equations read
\[
\begin{cases}
\dot{\mu}_{11}^{1}=-\frac{\omega}{2}\left(\mu_{11}^{2}+\mu_{21}^{1}+\mu_{12}^{1}\right),\qquad
\dot{\mu}_{11}^{2}=\frac{\omega}{2}\left(\mu_{11}^{1}-\mu_{21}^{2}-\mu_{12}^{2}\right)\\
\dot{\mu}_{12}^{1}=-\frac{\omega}{2}\left(\mu_{12}^{2}+\mu_{22}^{1}-\mu_{11}^{1}\right),\qquad
\dot{\mu}_{12}^{2}=\frac{\omega}{2}\left(\mu_{12}^{1}-\mu_{22}^{2}+\mu_{11}^{2}\right)\\
\dot{\mu}_{21}^{1}=-\frac{\omega}{2}\left(\mu_{21}^{2}-\mu_{11}^{1}+\mu_{22}^{1}\right),\qquad
\dot{\mu}_{21}^{2}=\frac{\omega}{2}\left(\mu_{21}^{1}+\mu_{11}^{2}-\mu_{22}^{2}\right)\\
\dot{\mu}_{22}^{1}=-\frac{\omega}{2}\left(\mu_{22}^{2}-\mu_{12}^{1}-\mu_{21}^{1}\right),\qquad
\dot{\mu}_{22}^{2}=\frac{\omega}{2}\left(\mu_{22}^{1}+\mu_{12}^{2}+\mu_{21}^{2}\right)\\
\end{cases}
\]
\end{lemma}

For the harmonic oscillator, define its auxiliary functions $A_\pm$ and $D_\pm$ by
\[
\begin{cases}
A_+^2+A_-^2=2\sqrt{2H}\\
A_+^2-A_-^2=2p\\
A_+A_-=\omega q\\
\end{cases},\qquad
\begin{cases}
D_+\doteq \frac{A_+}{2}(A_+^2-3A_-^2)\\
D_-\doteq \frac{A_-}{2}(3A_+^2-A_-^2)\\
\end{cases}
\]
Then one has the following
\begin{thm}
Let $C_{\be}\in\mathbb{R}$ ($\be=1,\ldots,8$) be arbitrary real--valued parameters,
$
M\=
\frac{\omega}{2}
\left(
\begin{smallmatrix}
0&-1\\
1&0
\end{smallmatrix}
\right)
$
and
\[
\begin{cases}
\mu_{11}^{1}(q,p)=C_5A_-+C_6A_++C_7D_-+C_8D_+\\
\mu_{12}^{1}(q,p)=C_1A_++C_2A_--C_7D_++C_8D_-\\
\mu_{21}^{1}(q,p)=-C_1A_+-C_2A_--C_3A_+-C_4A_--C_5A_++C_6A_--C_7D_++C_8D_-\\
\mu_{22}^{1}(q,p)=-C_3A_-+C_4A_+-C_7D_--C_8D_+\\
\mu_{11}^{2}(q,p)=C_3A_++C_4A_--C_7D_++C_8D_-\\
\mu_{12}^{2}(q,p)=C_1A_--C_2A_++C_3A_--C_4A_++C_5A_-+C_6A_+-C_7D_--C_8D_+\\
\mu_{21}^{2}(q,p)=-C_1A_-+C_2A_+-C_7D_--C_8D_+\\
\mu_{22}^{2}(q,p)=-C_5A_++C_6A_-+C_7D_+-C_8D_-\\
\end{cases}
\]
Then $(\mu,M)$ is a $2$-dimensional binary operadic Lax pair of the harmonic oscillator.
\end{thm}

\begin{proof}[Idea of proof]
Denote
\[
\begin{cases}
G_{\pm}^{\omega/2}&\doteq \dot{A}_{\pm}\pm\frac{\omega}{2}A_{\mp}\\
G_{\pm}^{3\omega/2}&\doteq \dot{D}_{\pm}\pm\frac{3\omega}{2}D_{\mp}\\
\end{cases}
\]
Define the matrix
\[
\Gamma
=(\Gamma_{\al}^{\be})\doteq\begin{pmatrix}
                                   0 & \hphantom{-}G_+^{\omega/2} & -G_+^{\omega/2} & 0 & 0 & \hphantom{-}G_-^{\omega/2} & -G_-^{\omega/2} & 0 \\
                                   0 & \hphantom{-}G_-^{\omega/2} & -G_-^{\omega/2} & 0 & 0 & -G_+^{\omega/2} & \hphantom{-}G_+^{\omega/2} & 0 \\
                                   0 & 0 & -G_+^{\omega/2} & -G_-^{\omega/2} & \hphantom{-}G_+^{\omega/2} & \hphantom{-}G_-^{\omega/2} & 0 & 0 \\
                                   0 & 0 & -G_-^{\omega/2} & \hphantom{-}G_+^{\omega/2} & \hphantom{-}G_-^{\omega/2} & -G_+^{\omega/2} & 0 & 0 \\
                                   G_-^{\omega/2} & 0 & -G_+^{\omega/2} & 0 & 0 & \hphantom{-}G_-^{\omega/2} & 0 & -G_+^{\omega/2} \\
                                   G_+^{\omega/2} & 0 & \hphantom{-}G_-^{\omega/2} & 0 & 0 & \hphantom{-}G_+^{\omega/2} & 0 & \hphantom{-}G_-^{\omega/2} \\
                                   G_-^{3\omega/2} & -G_+^{3\omega/2} & -G_+^{3\omega/2} & -G_-^{3\omega/2} & -G_+^{3\omega/2} & -G_-^{3\omega/2} & -G_-^{3\omega/2} & \hphantom{-}G_+^{3\omega/2} \\
                                   G_+^{3\omega/2} & \hphantom{-}G_-^{3\omega/2} & \hphantom{-}G_-^{3\omega/2} & -G_+^{3\omega/2} & \hphantom{-}G_-^{3\omega/2} & -G_+^{3\omega/2} & -G_+^{3\omega/2} & -G_-^{3\omega/2} \\
                                 \end{pmatrix}
\]
Then it follows from Lemma \ref{lemma:first} that the $2$-dimensional binary operadic Lax equations read
\[
C_{\be}\Gamma_{\al}^{\be}=0,\qquad \al=1,\ldots,8
\]
Since the parameters $C_\be$ are arbitrary, the latter constraints imply  $\Gamma=0$.
Thus one has to consider the following differential equations
\[
G_{\pm}^{\omega/2}=0=G_{\pm}^{3\omega/2}
\]
By direct calculations one can show that
\[
G_{\pm}^{\omega/2}=0
\qquad \Longleftrightarrow\qquad
\begin{cases}
\dot{p}=-\omega^{2}q\\
\dot{q}=p\\
\end{cases}\qquad \Longleftrightarrow\qquad
G_{\pm}^{3\omega/2}=0
\tag*{\qed}
\]
\renewcommand{\qed}{}
\end{proof}

\section*{Acknowledgement}
The research was in part supported by the Estonian Science Foundation, Grant 6912.
More expanded version of the present paper will be published in \cite{PV08}.

\end{document}